\documentclass[journal,twocolumn]{IEEEtran}

\usepackage{epsfig}
\usepackage{subfigure}
\usepackage{graphicx}
\usepackage{cite}

\usepackage{amsmath}
\usepackage{amssymb}

\def\BibTeX{{\rm B\kern-.05em{\sc i\kern-.025em b}\kern-.08em
    T\kern-.1667em\lower.7ex\hbox{E}\kern-.125emX}}

\newtheorem{claim}{Claim}
\newtheorem{corollary}{Corollary}
\newtheorem{theorem}{Theorem}

\newtheorem{lemma}{Lemma}

\newtheorem{remark}{Remark}

\title{Approximate Feedback Capacity of the Gaussian Multicast Channel}

\author{Changho Suh, Naveen Goela, Michael Gastpar \\
\thanks{C. Suh is with the Research Laboratory of Electronics at Massachusetts Institute of Technology, Cambridge, USA (e-mail: $\sf{chsuh@mit.edu}$).}
\thanks{N. Goela and M. Gastpar are with the School of Computer and Communication Sciences, Ecole Polytechnique F\'ed\'erale (EPFL), Lausanne, Switzerland (e-mail: $\mathsf{ \{naveen.goela, michael.gastpar \}@epfl.ch }$). They are also with the Department of Electrical Engineering and Computer Sciences, University of California, Berkeley, Berkeley, USA (Email: ${\sf \{ngoela,michael.gastpar\}@eecs.berkeley.edu}$)}
}

\begin{document}

\maketitle

\begin{abstract}
We characterize the capacity region to within $\log \left\{ 2(M-1) \right\}$ bits/s/Hz for the $M$-transmitter $K$-receiver Gaussian multicast channel with feedback where each receiver wishes to decode every message from the $M$ transmitters. Extending Cover-Leung's achievable scheme intended for $(M,K)=(2,1)$, we show that this generalized scheme achieves the cutset-based outer bound within $\log \left\{ 2(M-1) \right\}$ bits per transmitter for all channel parameters. In contrast to the capacity in the non-feedback case, the feedback capacity improves upon the naive intersection of the feedback capacities of $K$ individual multiple access channels. We find that feedback provides \emph{unbounded multiplicative} gain at high signal-to-noise ratios as was shown in the Gaussian interference channel. To complement the results, we establish the exact feedback capacity of the Avestimehr-Diggavi-Tse (ADT) deterministic model, from which we make the observation that feedback can also be beneficial for \emph{function computation}.
\end{abstract}
\begin{keywords}
ADT Deterministic Model, Feedback Capacity, Function Computation, Gaussian Multicast Channel
\end{keywords}

\section{Introduction}

While feedback plays a significant role in improving the reliability of communication systems~\cite{SK:it}, a traditional viewpoint on feedback capacity has been pessimistic over the past few decades. This is mainly due to Shannon's original result on feedback capacity which shows that feedback provides no increase in capacity for discrete memoryless point-to-point channels~\cite{shannon:it}. For multiple-access channels (MACs), feedback can increase the capacity~\cite{Gaarder:it}; however, the increase in capacity for Gaussian MACs is bounded by $1$ bit for all channel parameters~\cite{Ozarow:it}.

In contrast to these results, recent research shows that feedback provides more significant gain for communication over interference channels~\cite{Kramer:it02,SuhTse,GastparAmosYossefWigger}. Interestingly, the feedback gain is shown to be unbounded for certain channel parameters; i.e., the gap between the feedback and non-feedback capacities can be arbitrarily large as the signal-to-noise ratio ($\sf SNR$) of each link increases. One distinction of interference channels with respect to MACs is that each receiver decodes its desired message in the presence of undesired interfering signals. A natural question to ask is whether \emph{feedback} gain depends crucially on the presence of \emph{interference}.

In this paper, we make progress towards addressing this question. To isolate the interference issue, we start with a Gaussian MAC with two transmitters and feedback. We then modify the channel by adding additional receivers with identical message demands and feedback links from those receivers to the two transmitters. We call the new channel the two-transmitter, $K$-receiver Gaussian multicast channel with feedback.\footnote{In the non-feedback case, this channel is well known as the compound MAC~\cite{Maric:isit05,Simeone:it09}. However, this name is not appropriate in the feedback case. The compound MAC has a single physical receiver which can feed back only one of the possible candidates of the received signals experiencing different channel states. In our model, on the other hand, all of the received signals can be fed back. } Note that this channel does not pose any interference while still maintaining the many-to-many structure of interference channels. We present a coding scheme for this channel which generalizes Cover-Leung's scheme (intended for MACs)~\cite{Cover:it81}, and achieves rates within $1$ bit/s/Hz per transmitter of the cut-set outer bound for all channel parameters. We further extend our results for the case of $M$-transmitters and approximate the feedback capacity within $\log \left\{ 2(M-1) \right\}$ bits/s/Hz per transmitter. We find that feedback can provide multiplicative gain in the high-$\sf SNR$ regime, and that feedback is useful not only for mitigating interference~\cite{SuhTse}, but also for providing qualitatively-similar gains for channels with a many-to-many structure. In particular, we find that the feedback capacity region strictly enlarges the intersection of the feedback capacity regions of $K$ individual MACs. This is in contrast to the non-feedback case where the capacity region is simply the intersection of $K$ individual MAC capacity regions.

To complement our results on approximate feedback capacity, we establish the exact feedback capacity region of the Avestimehr-Diggavi-Tse (ADT) deterministic model. As a by-product, we also find that feedback increases the achievable rates for function computation. Specifically using a two-transmitter two-receiver example where each receiver wants to reconstruct a modulo-$2$ sum of two independent Bernoulli sources generated at the two transmitters, we demonstrate that feedback can increase the non-feedback computing capacity.

{\bf Related Work:} Feedback strategies for MACs were studied previously in~\cite{Cover:it81,Willems:it82,BrossLapidoth:it05,Venkataramanan:it11,Ozarow:it,Kramer:it02}. For the two-user case, Cover and Leung~\cite{Cover:it81} developed an achievable scheme that employs block Markov encoding and a decode-and-forward scheme. Willems~\cite{Willems:it82} proved the optimality of this scheme for a class of deterministic channels.
Ozarow~\cite{Ozarow:it} established the exact feedback capacity region using a different approach based on Schalkwijk-Kailath's scheme~\cite{SK:it}. Kramer developed more generic techniques of the approach and extended the result to include an arbitrary number of transmitters~\cite{Kramer:it02}. In the present paper, we generalize Cover-Leung's scheme to approximate the feedback capacity region of the $M$-transmitter $K$-receiver Gaussian multicast channel with an additional $(K-1)$ receivers as well as corresponding feedback links from those receivers to the $M$ transmitters.

The two-user compound MAC with conference encoders~\cite{Maric:isit05} or decoders~\cite{Simeone:it09} is also partially related to our work in the sense that dependence between the transmitted signals (or received signals) can be created through conferencing encoders (or decoders). However, we find a significant distinction. In the conferencing encoder problem~\cite{Maric:isit05,Willems:it83}, the capacity region is shown to be the intersection of the capacity regions of individual MACs. Similar behaviors follow for a class of conferencing decoder problems~\cite{Simeone:it09}. In contrast, we find that the feedback capacity region of our multicast channel enlarges the intersection of the feedback capacity regions of individual MACs.

Recently, Lim-Kim-El Gamal-Chung developed an achievable scheme for discrete memoryless networks~\cite{LimKimElGamalChung:it11}, and demonstrated the approximate optimality of their scheme for multi-source Gaussian multicast networks. Our feedback channel with unfolding can be cast into a multi-source Gaussian multicast network. However, we exploit the structure of our feedback channel to induce correlation between transmitters which leads to a tighter result. 

\section{Model}

\begin{figure}[t]
\begin{center}
{\epsfig{figure=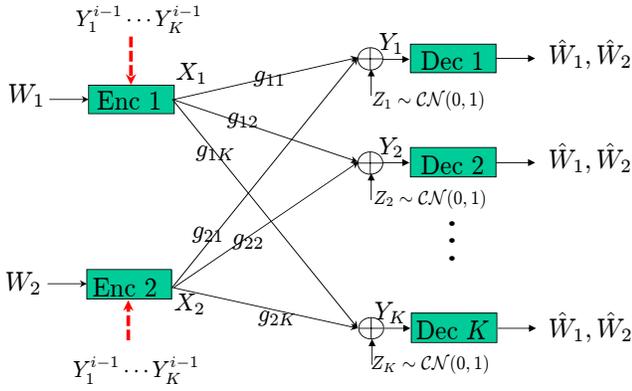, angle=0, width=0.46\textwidth}}
\end{center}
\caption{A Gaussian multicast channel with $M = 2$ transmitters and feedback from $K$ receivers.}
\label{fig:two-transmitterGaussianMulticast}
\end{figure}

We focus on the Gaussian multicast channel with $M = 2$ transmitters and $K$ receivers first. Section~\ref{sec:generalization} includes our results for $M > 2$. As shown in Fig.~\ref{fig:two-transmitterGaussianMulticast}, each receiver decodes all of the messages and is able to feed its received signal back to both transmitters. Without loss of generality, we normalize the transmit signal powers as $P_1 = P_2 = 1$ and channel noise powers as $Z_k \sim {\cal CN} (0,1)$ for all $k \in \{1,2, \ldots K\}$. Hence, the signal-to-noise ratio ($\sf SNR$) at each receiver captures the effect of the channel gains: ${\sf SNR}_{mk} \triangleq |g_{mk}|^2$, where $g_{mk} \in \mathbb{C}$ is the complex-valued channel gain from transmitter $m$ to receiver $k$.

Each transmitter $m \in \{1,2\}$ encodes an independent and uniformly distributed message $W_m \in \{1,2,\ldots,2^{NR_m}\}$. The encoded signal $X_{mi}$ of transmitter $m$ at time $i$ is a function of its own message and past feedback signals: $X_{mi} = f_{mi} \left(W_m, {Y}_{1}^{i-1}, \cdots, Y_{K}^{i-1}\right)$. We define ${Y}_k^{i-1} \triangleq \{Y_{kt}\}_{t=1}^{i-1}$ where $Y_{ki}$ is the received signal at receiver $k$ at time $i$. A rate pair $(R_1, R_2)$ is said to be achievable if there exists a family of codebooks subject to power constraints and corresponding encoding/decoding functions such that the average decoding error probabilities go to zero as the code length $N$ tends to infinity. The capacity region $\cal C$ is the closure of the set of the achievable rate pairs.


\section{Main Results}
\begin{theorem}[Inner Bound]
\label{thm:achievability}
The capacity region includes the set $\cal R$ of $(R_1,R_2)$ such that for $0 \leq \rho \leq 1$ and $\forall k$,
\begin{align}
\label{eq:achievablerate1}
R_1 &\leq \log \left( 1 + (1-\rho) \sum_{i=1}^{K} {\sf SNR}_{1i} \right ) \\
\label{eq:achievablerate2}
R_2 &\leq \log \left( 1 + (1-\rho) \sum_{i=1}^{K} {\sf SNR}_{2i} \right ) \\
\label{eq:achievablerate12}
R_1 + R_2 & \leq \log \left( 1 + {\sf SNR}_{1k} + {\sf SNR}_{2k} + 2 \rho \sqrt{{\sf SNR}_{1k} \cdot {\sf SNR}_{2k}}   \right).
\end{align}
\end{theorem}
\begin{proof}
See Section~\ref{sec:achievability}.
\end{proof}
\begin{remark}
\label{remark:largerthaninsection}
We compare this to the naive rate region which is the intersection of the feedback capacity regions of individual MACs:
\begin{align*}
{\cal R}_{\sf naive} = \bigcup_{ 0 \leq \rho \leq 1} \bigcap_{k=1}^{K} {\cal C}_k^{{\sf MAC}} (\rho), \end{align*}
where ${\cal C}_k^{{\sf MAC}} (\rho)$ denotes the feedback capacity region of the Gaussian MAC for receiver $k$, given $\rho$~\cite{Ozarow:it}. Note that the intersection constrains individual rate bounds, thus reducing the rate region. On the other hand, our rate region contains no such individual rate bounds, thus improving upon ${\cal R}_{\sf naive}$. This is in contrast to the nonfeedback case and the compound MAC case with encoders~\cite{Maric:isit05} (or decoders~\cite{Simeone:it09}),
where the capacity region is simply the intersection of individual MAC capacity regions. $\Box$
\end{remark}

\begin{theorem}[Outer Bound]
\label{thm:outerbound}
The capacity region is included by the set $\cal{\bar{C}}$ of $(R_1,R_2)$ such that for $0 \leq \rho \leq 1$ and $\forall k$,
\begin{align}
\label{eq:outerbound1}
R_1 &\leq \log \left( 1 + (1-\rho^2) \sum_{i=1}^{K} {\sf SNR}_{1i} \right ) \\
\label{eq:outerbound2}
R_2 &\leq \log \left( 1 + (1-\rho^2) \sum_{i=1}^{K} {\sf SNR}_{2i} \right ) \\
\label{eq:outerbound12}
R_1 + R_2 & \leq \log \left( 1 + {\sf SNR}_{1k} + {\sf SNR}_{2k} + 2 \rho \sqrt{{\sf SNR}_{1k} \cdot {\sf SNR}_{2k}}   \right).
\end{align}
\end{theorem}
\begin{proof}
See Section~\ref{sec:outerbound}.
\end{proof}
\begin{corollary}[One Bit Gap]
\label{cor:onebitgap}
The gap between the inner bound and outer bound regions given in Theorems~\ref{thm:achievability} and $\ref{thm:outerbound}$ is at most 1 bit/s/Hz/transmitter:
\begin{align*}
{\cal R} \subseteq {\cal C} \subseteq {\cal R} \oplus \left( [0,1] \times [0,1] \right).
\end{align*}
\end{corollary}
\begin{proof}
The proof is immediate. Let $\delta_1 =(\ref{eq:outerbound1}) - (\ref{eq:achievablerate1})$. Similarly we define $\delta_2$ and $\delta_{12}$. Straightforward computation then gives
$\delta_1 \leq \log (1 + \rho) \leq 1$.
Similarly, we get $\delta_2 \leq 1$ and $\delta_{12} =0$. This completes the proof.
\end{proof}

\begin{figure}[t]
\begin{center}
{\epsfig{figure=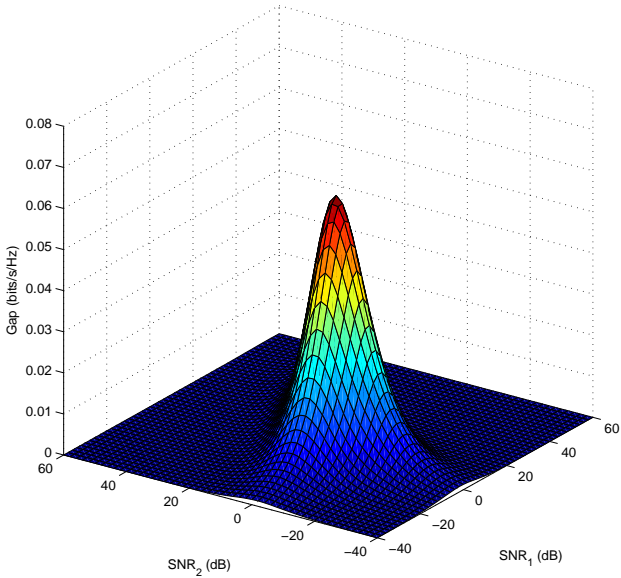, angle=0, width=0.5\textwidth}}
\end{center}
\caption{The gap between the symmetric-rate inner and outer bounds for a two-receiver symmetric channel setting: ${\sf SNR}_1 :={\sf SNR}_{11} ={\sf SNR}_{22}$ and ${\sf SNR}_2:= {\sf SNR}_{12} ={\sf SNR}_{21}$
} \label{fig:gap}
\end{figure}

\begin{remark}
\label{remark:onebitgap}
Fig.~\ref{fig:gap} shows a numerical result of the inner-and-upper bound gap for the symmetric capacity, denoted by $C_{\sf sym} = \sup \left \{R:  (R,R) \in {\cal C} \right \}$. For illustrative purpose, we consider a two-receiver symmetric channel setting where ${\sf SNR}_1 :={\sf SNR}_{11} ={\sf SNR}_{22}$ and ${\sf SNR}_2:= {\sf SNR}_{12} ={\sf SNR}_{21}$. While the worst-case gap is 1 bit due to the coarse analysis in Corollary~\ref{cor:onebitgap}, the actual gap is upper-bounded by approximately $0.08$ over a wide range of channel parameters. This suggests that a refined analysis could lead to an even smaller gap. For instance, in the high-${\sf SNR}$ regime, we obtain the asymptotic symmetric capacity as follows. $\Box$
\end{remark}
\begin{corollary}
\label{cor:asymptoticsymcapacity}
For a two-receiver symmetric channel setting, the symmetric capacity at the high $\sf SNR$ regime is
\begin{align}
\label{eq:aymptoticsymcapacity}
C_{\sf sym} \approx \frac{1}{2} \log \left (  {\sf SNR}_1 + {\sf SNR}_2 + 2 \sqrt{{\sf SNR}_1 \cdot {\sf SNR}_2} \right ).
\end{align}
\end{corollary}
\begin{proof}
Due to the high-$\sf SNR$ assumption, it follows that the optimal correlation coefficients for the inner and upper bounds are $\rho^*_{\sf in} \approx \rho_{\sf out}^{*2} \approx 1 - \frac{\sqrt{{\sf SNR}_1 + {\sf SNR}_2 + 2 \sqrt{{\sf SNR}_1 \cdot {\sf SNR}_2}}   }{{\sf SNR}_1 + {\sf SNR}_2 }$ respectively, resulting in the matching inner and upper bound as (\ref{eq:aymptoticsymcapacity}).
\end{proof}

{\bf Feedback Gain:}
From Theorems~\ref{thm:achievability} and $\ref{thm:outerbound}$, we can see that feedback can provide a significant capacity increase as was shown in the Gaussian interference channel~\cite{SuhTse}.
\begin{figure}[t]
\begin{center}
{\epsfig{figure=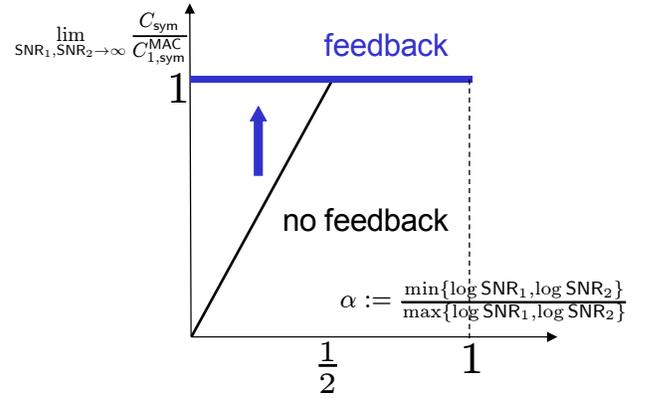, angle=0, width=0.45\textwidth}}
\end{center}
\caption{Feedback gain for a two-receiver symmetric channel setting as in Fig.~\ref{fig:gap}.
Note that feedback provides unbounded multiplicative gain when ${\sf SNR}_1$ is far apart from ${\sf SNR}_2$.
} \label{fig:feedbackgain}
\end{figure}
To see this clearly, let us consider the two-receiver symmetric channel setting as above. Fig.~\ref{fig:feedbackgain} plots the high-$\sf SNR$-regime symmetric capacity normalized by the MAC symmetric capacity for Rx 1, denoted by $C_{1, {\sf sym}}^{\sf {\sf MAC}} = \frac{1}{2} \log \left( 1+ {\sf SNR}_1 + {\sf SNR}_2 \right)$.  Here we use $\alpha:=\frac{ \min \{\log {\sf SNR}_{1}, \log {\sf SNR}_{2}\}}{ \max \{ \log {\sf SNR}_{1},\log {\sf SNR}_{2} \}}$ for $x$-axis to indicate a signal strength difference between ${\sf SNR}_{1}$ and ${\sf SNR}_{2}$.
Note that the symmetric nonfeedback capacity is simply the intersection of individual MAC capacities:
\begin{align*}
C_{\sf sym}^{\sf NO} = \min \left\{ \min_{i=1,2} \log ( 1+ {\sf SNR}_i ), \frac{1}{2} \log (1 + {\sf SNR}_1 + {\sf SNR}_2) \right\}.
\end{align*}
Note that the gap between $C_{\sf sym}^{\sf NO}$ and $C_{1,{\sf sym}}^{{\sf MAC}}$ can be arbitrarily large when ${\sf SNR}_1$ and ${\sf SNR}_2$ are far apart, i.e., $\alpha \leq \frac{1}{2}$. On the other hand, the symmetric feedback capacity is asymptotically the same as if there were only one receiver. As a result, feedback provides multiplicative gain for the regime of $\alpha \leq \frac{1}{2}$. In Section~\ref{sec:achievability}, we will provide an intuition behind this gain while describing an achievable scheme. $\Box$

\section{Gaussian Channel}

\subsection{Achievability: Proof of Theorem~\ref{thm:achievability}}
\label{sec:achievability}

{\bf Motivating Example (Fig.~\ref{fig:example}):}
To develop an achievable scheme for the Gaussian channel, we utilize the ADT deterministic model~\cite{Salman:IT11} illustrated in Fig.~\ref{fig:ADT} as an intermediate yet insightful model. The ADT multicast channel with $M$ transmitters and $K$ receivers is characterized by $MK$ values: $n_{mk}, \; 1 \leq m \leq M,  1 \leq k \leq K$ where $n_{mk}$ indicates the number of signal bit levels from transmitter $m$ to receiver $k$. These values correspond to the channel gains of the Gaussian channel in dB scale: $n_{mk} = \lfloor \log {\sf SNR}_{mk} \rfloor$.
See~\cite{Salman:IT11} for explicit details.

We first explain an achievable scheme for a particular ADT model example, illustrated in Fig.~\ref{fig:example}. Specifically, we show how to achieve a $(1.5,1.5)$ rate-pair with feedback. As will be seen in Theorem~\ref{theorem:ADT}, the feedback capacity region is given by $R_1 + R_2 \leq 3$. Extrapolating from this example, we later make observations leading to a generic achievable scheme.

\begin{figure}[t]
\begin{center}
{\epsfig{figure=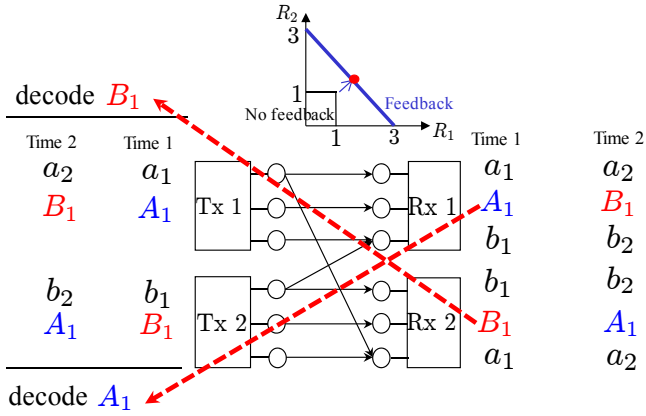, angle=0, width=0.47\textwidth}}
\end{center}
\caption{Motivating example: An achievable scheme for a $(1.5,1.5)$ rate-pair.} \label{fig:example}
\end{figure}

In the nonfeedback case, transmitter 1 can send only one bit $a_1$ through the top level, since the mincut between transmitter 1 and receiver 2 is limited by 1. Similarly transmitter 2 can send only one bit, say $b_1$. However, feedback provides more options to route by creating additional paths, e.g., $[Tx 1 \rightarrow Rx 1 \rightarrow \textrm{feedback} \rightarrow Tx 2 \rightarrow Rx 2]$. This additional path enables an increase over the nonfeedback rate. Transmitter 1 squeezes one more bit $A_1$ in the second level. Similarly transmitter 2 squeezes $B_1$ in its own second level. Receiver 1 then gets $A_1$, while receiver 2 does not. Similarly $B_1$ is received only at receiver 2.
We will show that these $A_1$ and $B_1$ can also be delivered to the other receivers with the help of feedback. At the beginning of time 2, transmitter 1 can decode $B_1$ with feedback. Similarly transmitter 2 can decode $A_1$. In time 2, transmitters 1 and 2 start with sending their own fresh information $a_2$ and $b_2$ on the top levels respectively. Now the idea is that transmitter 1 forwards the fed back $B_1$ using the second level. Note that this transmission allows receiver 1 to obtain $B_1$ without causing any harm to the transmission of $(a_2,b_2)$. Similarly transmitter 2 can deliver $A_1$ to receiver 2. Therefore, during the two time slots, transmitters 1 and 2 can deliver $(a_1,a_2,A_1)$ and $(b_1,b_2,B_1)$ respectively to both receivers, thus achieving $(1.5,1.5)$.

\begin{remark}
The gain comes from the fact that feedback creates alternative paths to provide \emph{routing} gain. In fact, this gain was already observed by~\cite{SuhTse} in the context of two-user strong interference channels where $n_{12} \geq n_{11}$ and $n_{21} \geq n_{22}$ in the ADT model. However in~\cite{SuhTse}, this routing gain does not appear in the weak interference regime such as $(n_{12}=1 < n_{11}=3, n_{21}=1 < n_{22}=3)$. On the other hand, in our multicast channel, we can see this routing gain even when cross links are weaker than direct links. $\Box$
\end{remark}

This example leads us to make two observations. First, feedback enables each transmitter to decode the other transmitter's information and then forwards this in the next time slot.
Second, the transmitted signals in time 2 can be correlated with the previously-sent information.
This motivates us to employ the decoding-and-forward and block Markov encoding schemes. In fact, an achievable scheme combining these two ideas was developed by Cover-Leung~\cite{Cover:it81} in the context of the two-user discrete memoryless MAC with feedback. In this paper, we generalize this scheme to the multiple-receiver case, thereby obtaining an approximate capacity region within a provably small gap. As for a decoding operation, we employ backward decoding~\cite{Willems:it85}.

Here is the outline of achievability. We employ block Markov encoding with a total size $B$ of blocks. In block 1, each transmitter sends its own information. In block 2, with feedback, each transmitter decodes the other user's information (sent in block 1). The two previously-sent messages are then available at each transmitter. Conditioning on these two messages, each transmitter generates its own fresh message and then sends a corresponding codeword. Each transmitter repeats this procedure until block $B-1$. In the last block $B$, to facilitate backward decoding, each transmitter sends a predetermined message. Each receiver waits until a total of $B$ blocks have been received and then performs backward decoding.

The achievable scheme outlined above is broadly applicable and not limited to the Gaussian channel. We characterize an achievable rate region for discrete mememoryless multicast channels in Lemma~\ref{lemma:achievablerateregion} and then choose an appropriate joint distribution to obtain the desired result. The generic coding scheme is also applicable to the ADT deterministic model and details will be presented in Section~\ref{sec:deterministic_channels}.

\begin{lemma}
\label{lemma:achievablerateregion}
The feedback capacity region of the two-transmitter $K$-receiver discrete memoryless multicast channel includes the set of $(R_1,R_2)$ such that
\begin{align}
R_1 &\leq I(X_1; Y_1,\cdots, Y_K | X_2, U) \\
R_2 &\leq I(X_2; Y_1,\cdots, Y_K | X_1, U) \\
R_1 + R_2 & \leq I(X_1,X_2; Y_k), \forall k
\end{align}
over all joint distributions $p(u) p(x_1|u) p(x_2|u)$. Here $U$ is a discrete random variable which takes on values in the set ${\cal U}$ where $|{\cal U}| \leq \min \left\{ |{\cal X}_1| |{\cal X}_2|, |{\cal Y}_1|, \cdots, |{\cal Y}_K| \right\} + 2$.
\end{lemma}
\begin{proof}
See Appendix~\ref{appendix:proofoflemma}.
\end{proof}
We now choose the following Gaussian input distribution to complete the proof: $\forall m=1,2$,
\begin{align}
U \sim {\cal CN} (0,\rho); \tilde{X}_m \sim {\cal CN}(0,1-\rho),
\end{align}
where $X_m = U+\tilde{X}_m$ and $(U, \tilde{X}_1,\tilde{X}_2)$ are independent. Straightforward computation then gives (\ref{eq:achievablerate1})-(\ref{eq:achievablerate12}). This completes the proof.

\subsection{Outer Bound: Proof of Theorem~\ref{thm:outerbound}}
\label{sec:outerbound}

By symmetry, it suffices to prove the bounds of (\ref{eq:outerbound1}) and (\ref{eq:outerbound12}). These bounds are based on standard cut-set arguments. 

Assume that the covariance between $X_1$ and $X_2$ is $E [X_1 X_2^*]=\rho$. Starting with Fano's inequality,
\begin{align*}
\begin{split}
&N(R_1 - \epsilon_N) \leq I(W_1;Y_1^{N}, \cdots, Y_K^N, W_2) \\
& \overset{(a)}{=} \sum h(Y_{1i}, \cdots, {Y}_{Ki} | W_2, Y_1^{i-1}, \cdots, {Y}_K^{i-1}, X_{2i}) \\
& \qquad - h(Y_{1i}, \cdots, {Y}_{Ki} | W_1, W_2, Y_1^{i-1}, \cdots, {Y}_K^{i-1}, X_{2i}, X_{1i}) \\
& \overset{(b)}{\leq} \sum [h(Y_{1i},\cdots,Y_{Ki} | X_{2i}) - h(Z_{1i},\cdots,Z_{Ki})]  \\
& \overset{(c)}{\leq} N \log \left( 1 + (1-|\rho|^2) \sum_{k=1}^{K} {\sf SNR}_{1k} \right )
\end{split}
\end{align*}
where $(a)$ follows from the fact that $W_1$ is independent of $W_2$, and $X_{mi}$ is a function of $(W_m,{Y}_1^{i-1},\cdots,Y_K^{i-1})$; $(b)$ follows from the fact that conditioning reduces entropy and channel is memoryless; and $(c)$ follows from the fact that $|K_{Y_1,\cdots,Y_K |X_2}| \leq  1+  (1 - |\rho|^2)\sum_{k} {\sf SNR}_{1k}$. If $R_1$ is achievable, then $\epsilon_N \rightarrow 0$ as $N$ tends to infinity. Therefore, we get the desired bound.

For the sum-rate outer bound,
\begin{align*}
\begin{split}
&N(R_1 + R_2- \epsilon_N) \leq I(W_1, W_2;Y_1^{N}) \\
& \overset{(a)}{\leq} \sum [h(Y_{1i}) - h(Y_{1i}| W_1, W_2, Y_1^{i-1}, \cdots, Y_K^{i-1}, X_{1i}, X_{2i})]  \\
& \overset{(b)}{=} \sum [h(Y_{1i}) - h(Z_{1i})]  \\
& \overset{(c)}{\leq}  N \log \left( 1 + {\sf SNR}_{11} + {\sf SNR}_{21} + 2 |\rho | \sqrt{{\sf SNR}_{11} \cdot {\sf SNR}_{21}} \right )
\end{split}
\end{align*}
where $(a)$ follows from the fact that conditioning reduces entropy; $(b)$ follows from the memoryless property of channels; and $(c)$ follows from the fact that $|K_{Y_1}| \leq 1 + {\sf SNR}_{11} + {\sf SNR}_{21} + 2 |\rho | \sqrt{{\sf SNR}_{11} \cdot {\sf SNR}_{21}}$.

\subsection{Generalization to $M$-transmitter Case}
\label{sec:generalization}

\begin{theorem}[Inner Bound]
\label{thm:achievability_Mtx}
The feedback capacity region of the $M$-transmitter $K$-receiver Gaussian multicast channel includes the set ${\cal R}_M$ of $(R_1,\cdots, R_M)$ such that for $0 \leq \rho \leq 1$, $\forall {\cal S}  \subsetneq \{1,\cdots, M \}$ and $\forall k$,
\begin{align}
\label{eq:achievablerate1_Mtx}
\sum_{m \in {\cal S}} R_m &\leq \log \left|
I_{K} + (1-\rho) G_{\cal S} G_{\cal S}^* \right | \\
\label{eq:achievablerate12_Mtx}
\sum_{m=1}^{M} R_m & \leq \log \left (1 + \sum_{m=1}^{M} {\sf SNR}_{mk}  + \sum_{m \neq n}\rho \sqrt{
{\sf SNR}_{mk} \cdot {\sf SNR}_{nk} } \right )
\end{align}
where $G_{\cal S}$ is such that
\begin{align*}
Y = G_{\cal S} X_{\cal S} + G_{{\cal S}^C} X_{{\cal S}^C} +   Z.
\end{align*}
Here $Y:=[Y_{1},\cdots,Y_K]^t \in \mathbb{C}^{K} $; $X_{\cal S}:=[X_m]^t \in \mathbb{C}^{|\cal S|}, m \in {\cal S}$; and $Z:=[Z_{1},\cdots,Z_K]^t$.
\end{theorem}
\begin{proof}
We first generalize Lemma~\ref{lemma:achievablerateregion} as follows.
\begin{lemma}
\label{lemma:generalachievable}
The feedback capacity region of the $M$-transmitter $K$-receiver discrete memoryless multicast channel includes the set of $(R_1,\cdots,R_M)$ such that $\forall {\cal S}  \subsetneq \{1,\cdots, M \}$ and $\forall k$,
\begin{align*}
\sum_{m \in {\cal S}} R_m &\leq I(X_{\cal S}; Y_1,\cdots, Y_K | X_{{\cal S}^c}, U) \\
\sum_{m=1}^{M} R_m & \leq I (X_1,\cdots, X_{M}; Y_k )
\end{align*}
over all joint distributions $p(u) \prod_{m=1}^{M} p(x_m|u)$. Here $U$ is a discrete random variable which takes on values in the set ${\cal U}$ where $|{\cal U}| \leq \min \left\{ |{\cal X}_1| |{\cal X}_2| \cdots |{\cal X}_M|, |{\cal Y}_1|, \cdots, |{\cal Y}_K| \right\} + 2$.
\end{lemma}
\begin{proof}
For $M>2$, a multitude of auxiliary random variables can be incorporated to capture correlation between many transmitter pairs. For simplicity, however, we consider a natural extension of the two-transmitter case which includes only one auxiliary random variable.
The only distinction is that with feedback, each transmitter decodes all of the messages of the other transmitters, and generates its new message and a corresponding codeword, conditioned on all of these decoded messages. This induces a multitude of constraints on the rate region. To avoid significant overlaps, we omit the detailed proof.
\end{proof}
We now choose the following Gaussian input distribution to complete the proof: $\forall m=1,\cdots,M$,
\begin{align}
U \sim {\cal CN} (0,\rho); \tilde{X}_m \sim {\cal CN}(0,1-\rho),
\end{align}
where $X_m = U+\tilde{X}_m$ and $(U, \tilde{X}_1, \cdots, \tilde{X}_M)$ are independent. Straightforward computation then gives (\ref{eq:achievablerate1_Mtx})-(\ref{eq:achievablerate12_Mtx}). This completes the proof.
\end{proof}

\begin{theorem}[Outer Bound]
\label{thm:outerbound_Mtx}
The feedback capacity region of the $M$-transmitter $K$-receiver Gaussian multicast channel is included by the set ${\cal{\bar{C}}}_{M}$ of $(R_1,\cdots, R_M)$ such that $\forall K_{X}:=E[XX^*] \succeq {\bf 0}$, $\forall {\cal S}  \subsetneq \{1,\cdots, M \}$ and $\forall k$,
\begin{align}
\label{eq:outerbound1_Mtx}
\sum_{m \in {\cal S}} R_m &\leq \log \left|
I_{K} + G_{\cal S} K_{{\cal S} | {\cal S}^c } G_{\cal S}^* \right | \\
\label{eq:outerbound12_Mtx}
\sum_{m=1}^{M} R_m & \leq \log \left (1 + \sum_{m=1}^{M} {\sf SNR}_{mk}  + \sum_{m \neq n}\rho_{mn} \sqrt{
{\sf SNR}_{mk} \cdot {\sf SNR}_{nk} } \right )
\end{align}
where $K_{{\cal S} | {\cal S}^c }$ denotes the conditional covariance matrix of $X_{\cal S}$ given $X_{ {\cal S}^c }$ and $\rho_{mn}:= \left| [K_{X}]_{mn} \right|$.
\end{theorem}
\begin{proof}
As before, the proof of the outer bounds are based on the standard cutset argument. Hence, we omit the detailed proofs.
\end{proof}
\begin{corollary}[Constant Gap]
\label{cor:constantgap}
The gap between the inner bound and outer bound regions given in Theorems~\ref{thm:achievability_Mtx} and $\ref{thm:outerbound_Mtx}$ is upper-bounded by $\Delta:= \log \left\{ 2(M-1) \right\}$ bits/s/Hz/transmitter:
\begin{align*}
{\cal R}_M \subseteq {\cal C}_M \subseteq {\cal R}_M \oplus \left( [0, \Delta] \times \cdots \times [0,\Delta] \right).
\end{align*}
\end{corollary}
\begin{proof}
See Appendix~\ref{appendix:constantgap}.
\end{proof}

\section{Deterministic Channel}
\label{sec:deterministic_channels}

The ADT model was developed as a method of analysis to approximate the feedback capacity region of the Gaussian multicast channel. In this section, we find the exact feedback capacity region of the deterministic channel.
\begin{figure}[t]
\begin{center}
{\epsfig{figure=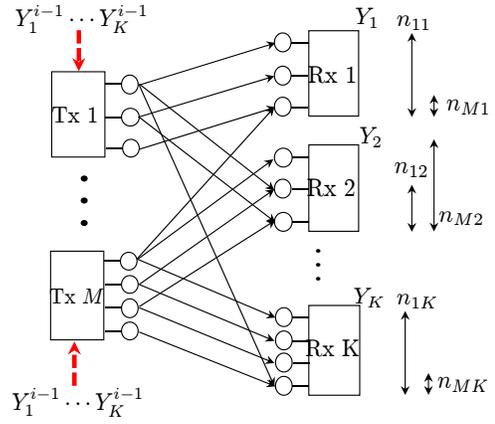, angle=0, width=0.35\textwidth}}
\end{center}
\caption{An ADT multicast channel with $M$ transmitters and $K$ receivers providing noiseless feedback.} \label{fig:ADT}
\end{figure}

\begin{theorem}
\label{theorem:ADT}
The feedback capacity region of the $M$-transmitter $K$-receiver ADT multicast channel is the set of $(R_1,\cdots, R_M)$ such that $\forall {\cal S}  \subsetneq \{1,\cdots, M \}$ and $\forall k$,
\begin{align}
\label{eq:ADTbound1}
\sum_{m \in {\cal S}} R_m &\leq {\sf rank} (G_{\cal S}) \\
\label{eq:ADTbound2}
\sum_{m=1}^{M} R_m & \leq \max \left\{ n_{1k}, \cdots, n_{Mk} \right\},
\end{align}
where $G_{\cal S}$ is such that $Y = G_{\cal S} X_{\cal S} + G_{{\cal S}^c} X_{{\cal S}^c}$.
\end{theorem}
\begin{proof}
The achievability proof is immediate due to Lemma~\ref{lemma:generalachievable}.
The achievable region is maximized when $U=\varnothing$ and $(X_1,\cdots,X_M)$ are uniformly distributed and independent. Appendix~\ref{appendix:converseproof_ADT} contains the converse proof.
\end{proof}


\section{Function Computation}
\label{sec:function_computation}

As a by-product of Theorem~\ref{theorem:ADT}, we can find an interesting role of feedback for other communication scenarios such as computation in networks. To see this, consider an $(M,K) = (2,2)$ ADT multicast channel with feedback and parameters $n_{11} = n_{22} = 3$ and $n_{12} = n_{21} = 1$ (see Fig.~\ref{fig:example}). Suppose that both receivers wish to compute the same function of modulo-2 sums of two independent Bernoulli sources ($S_1, S_2$) generated at the two transmitters. The computing rate for decoding $S_1 \oplus S_2$ at all receivers is denoted $R_{\sf comp}$. Without feedback, the following cut-set based argument provides a bound on $R_{\sf comp}$:
\begin{align*}
N&(R_{\sf comp} - \epsilon_{N}) \leq I(S_1 \oplus S_2;Y_1^{N}) \\
&\leq I(S_1 \oplus S_2; Y_1^{N}, S_1) = I(S_1 \oplus S_2; Y_1^{N}| S_1,X_1^{N}) \\
&\leq H(Y_1^{N}|S_1, X_1^{N}) \leq \sum H(Y_{1i}|X_{1i}),
\end{align*}
where the equality follows from the fact that $S_1 \oplus S_2$ is independent of $S_1$. For the particular ADT example, $H(Y_1|X_1) \leq 1$ and $H(Y_2|X_2) \leq 1$, from which  $R_{\sf comp} \leq 1$. On the other hand, the example in Fig.~\ref{fig:example} shows the achievability of $(\frac{3}{2},\frac{3}{2})$, thus yielding $R_{\sf comp}^{\sf FB} \geq \frac{3}{2}$. Therefore, feedback can increase rates for computation. Our future work is to extend this example to larger classes of networks.


\section{Conclusion}
\label{sec:conclusion}

We established the feedback capacity region of the Gaussian multicast channel with $M$ transmitters and $K$ receivers to within $\log \left\{ 2 (M-1) \right\}$ bits/s/Hz per transmitter of the cutset bound universally over all channel parameters. 
We characterized the exact feedback capacity region of the ADT model, and observed a feedback gain for function computation.

Our future work is along several new directions: (1) Improving our coding scheme based on Cover-Leung to incorporate ideas from~\cite{BrossLapidoth:it05, Venkataramanan:it11,Ozarow:it,Massimo:LQG}; (2) Extending to more realistic scenarios where feedback is offered through rate-limited bit-piped links~\cite{AlirezaSuhAves} or a corresponding backward channel~\cite{SuhTse:isit12}; (3) Exploring the role of feedback for function computation.

\appendices

\section{Proof of Lemma~\ref{lemma:achievablerateregion}}
\label{appendix:proofoflemma}

\textbf{Codebook Generation:} Fix a joint distribution $p(u) p(x_1|u)p(x_2|u)$.
First generate $2^{N(R_{1} +R_{2}) }$ independent codewords $u^N(j,l)$, $j \in \{1, \cdots, 2^{NR_{1}} \}$, $l \in \{1, \cdots, 2^{NR_{2}} \}$, according to $\prod_{i=1}^{N} p(u_i)$. For each codeword $u^N(j,l)$, encoder 1 generates $2^{NR_{1}}$ independent codewords $x_1^{N}((j,l),s)$, $s \in \{1, \cdots, 2^{NR_{1}} \}$, according to $\prod_{i=1}^{N} p(x_{1i}|u_i)$. Similarly, for each codeword $u^{N}(j,l)$, encoder 2 generates $2^{NR_{2}}$ independent codewords $x_2^{N}((j,l),q)$, $q \in \{1, \cdots, 2^{NR_{2}} \}$, according to $\prod_{i=1}^{N} p(x_{2i}|u_i)$.


\textbf{Encoding and Decoding:} We employ block Markov encoding with a total size $B$ of blocks. Focus on the $b$th block transmission. With feedback $(y_1^{N,(b-1)}, \cdots, y_K^{N,(b-1)})$, transmitter 1 tries to decode the message $\hat{w}_{2}^{(b-1)} = \hat{q}$ (sent from transmitter 2 in the $(b-1)$th block). In other words, we find the unique $\hat{q}$ such that
\begin{align*}
\begin{split}
&\left( u^N \left( w_{1}^{(b-2)}, \hat{w}_{2}^{(b-2)} \right), x_1^N \left( ( w_{1}^{(b-2)}, \hat{w}_{2}^{(b-2)}), w_{1}^{(b-1)} \right), \right. \\
& \left. \;\; x_2^N \left( (w_{1}^{(b-2)}, \hat{w}_{2}^{(b-2)}), \hat{q} \right) , y_1^{N,(b-1)}, \cdots, y_K^{N,(b-1)} \right) \in A_{\epsilon}^{(N)},
\end{split}
\end{align*}
where $A_{\epsilon}^{(N)}$ indicates the set of jointly typical sequences.
Note that transmitter 1 already knows its own messages $(w_{1}^{(b-2)}, w_{1}^{(b-1)})$. We assume that $\hat{w}_{2}^{(b-2)}$ is correctly decoded from the previous block $(b-1)$. The decoding error occurs if one of two events happens: (1) there is no typical sequence; (2) there is another $\hat{w}_{2}^{(b-1)}$ such that it is a typical sequence. By AEP, the first error probability becomes negligible as $N$ tends to infinity. By the packing lemma in~\cite{ElGamalKim:Book,CoverThomas}, the second error probability becomes arbitrarily small (as $N$ tends to infinity) if
\begin{align}
\label{eq-R2c-constraint}
R_{2} \leq I(X_2;Y_1,\cdots,Y_K|X_1,U).
\end{align}
Based on $(w_{1}^{(b-1)},\hat{w}_{2}^{(b-1)})$, transmitter 1 generates a new message $w_{1}^{(b)}$ and then sends $x_1^N \left( (w_{1}^{(b-1)}, \hat{w}_{2}^{(b-1)} ),w_{1}^{(b)} \right)$. Similarly transmitter 2 decodes $\hat{w}_{1}^{(b-1)}$, generates $w_{2}^{(b)}$ and then sends $x_2^N \left((\hat{w}_{1}^{(b-1)}, w_{2}^{(b-1)} ),w_{2}^{(b)} \right)$.

Each receiver waits until total $B$ blocks have been received and then does backward decoding.
Notice that a block index $b$ starts from the last $B$ and ends to $1$. For block $b$, receiver $k$ finds the unique pair  $(\hat{j},\hat{l})$ such that
\begin{align*}
\begin{split}
&\left( u^N \left( \hat{j}, \hat{l} \right), x_1^N \left( ( \hat{j}, \hat{l}) , \hat{w}_{1}^{(b)} \right), \right. \\
&\left. \qquad  x_2^N \left( (\hat{j}, \hat{l}), \hat{w}_{2}^{(b)} \right) , y_k^{N,(b)}  \right) \in A_{\epsilon}^{(N)},
\end{split}
\end{align*}
where we assumed that a pair of messages $(\hat{w}_{1}^{(b)}, \hat{w}_{2}^{(b)})$ was successively decoded from block $(b+1)$. Similarly other receivers follow the same decoding procedure.

\textbf{Error Probability:} By symmetry, we consider the probability of error only for block $b$ at receiver $k$. We assume that $(w_{1}^{(b-1)},w_{2}^{(b-1)}) = (1,1)$ was sent through block $(b-1)$ and block $b$; and there was no backward decoding error from block $B$ to $(b+1)$, i.e., $(\hat{w}_{1}^{(b)}, \hat{w}_{2}^{(b)})$ are successfully decoded.

Define an event:
\begin{align*}
E_{jl} = \left\{ \left( u^{N}(j,l), x_1^{N}((j,l),\hat{w}_{1}^{(b)}), \right. \right. \\
\left. \left. \;\; x_2^{N}((j,l),\hat{w}_{2}^{(b)}), y_k^{N,(b)} \right)  \in A_{\epsilon}^{(N)} \right\}.
\end{align*}
By AEP, the first type of error becomes negligible. Hence, we focus only on the second type of error. Using the union bound, we get
\begin{align}
\begin{split}
\label{eq-errorprobability}
& \textrm{Pr}  \left( \bigcup_{(j,l)  \neq (1,1) } E_{jl}   \right) \leq \sum_{j \neq 1, l \neq 1}\textrm{Pr}(E_{jl} ) \\
& + \sum_{j \neq 1, l = 1}\textrm{Pr}(E_{j1} ) + \sum_{j = 1, l \neq 1}\textrm{Pr}(E_{1l} )  \\
& \leq 3 \cdot 2^{N(R_{1}+R_{2} - I(U,X_1,X_2;Y_k) + 3 \epsilon)}.
\end{split}
\end{align}
Here note that $( j \neq 1, l \neq 1)$ is the worst case, dominating the other two cases. The number $3$ in the second inequality reflects all of these cases. Hence, the error probability can be made arbitrarily small if
\begin{align}
\label{eq-R1R2constraint}
    R_{1}+R_{2} \leq I(X_1, X_2; Y_k), \forall k.
\end{align}
From (\ref{eq-R2c-constraint}) and (\ref{eq-R1R2constraint}), we complete the proof.

\section{Proof of Corollary~\ref{cor:constantgap}}
\label{appendix:constantgap}

Let $\delta_{\cal S} =(\ref{eq:outerbound1_Mtx}) - (\ref{eq:achievablerate1_Mtx})$. Set $\rho = \min_{m,n} \rho_{ij}$.
We then get
\begin{align*}
&\delta_{\cal S}  = \log \left|
I_{K} + G_{\cal S} K_{{\cal S} | {\cal S}^c } G_{\cal S}^* \right | -
\log \left|
I_{K} + (1-\rho) G_{\cal S} G_{\cal S}^* \right |  \\
&\overset{(a)}{\leq} \log \left|
I_K + |{\cal S}|(1- \rho^2) G_{\cal S} G_{\cal S}^* \right | - \log \left|
I_{K} + (1-\rho) G_{\cal S} G_{\cal S}^* \right | \\
&\overset{(b)}{=} \log \left|
I_{|{\cal S}|} + |{\cal S}|(1- \rho^2) G_{\cal S}^* G_{\cal S} \right | - \log \left|
I_{|{\cal S}|} + (1-\rho) G_{\cal S}^* G_{\cal S} \right | \\
&\leq \log \left| (1+\rho) |{\cal S}| I_{|{\cal S}|}  \right | \\
&= |{\cal S}| \log \left| (1+\rho) |{\cal S}|  \right | \\
&\leq |{\cal S}| \log \left\{ 2 (M-1) \right \}
\end{align*}
where $(a)$ follows from Claim 1 (see below); and $(b)$ follows from the determinant identity $\left|I_m + A B \right| = \left|I_n + BA \right|$. Therefore, the gap per transmitter is upper-bounded by
\begin{align}
\frac{\delta_{\cal S}}{|{\cal S}|} \leq \log \left\{ 2 (M-1) \right\}.
\end{align}

\begin{claim}
$K_{{\cal S} | {\cal S}^c } \preceq  |{\cal S}|(1-\rho^2)  I_{|{\cal S}|}.$
\end{claim}
\begin{proof}
Starting with the fact that any covariance matrix is positive semidefinite, we get
\begin{align*}
K_{{\cal S} | {\cal S}^c } & \preceq {\sf trace} \left( K_{{\cal S} | {\cal S}^c } \right) I_{|{\cal S}|} \\
& =\sum_{m \in {\cal S}} \left \{ K_{ X_m | {\cal S}^c } \right \} I_{|{\cal S}|}  \overset{(a)}{\preceq} \sum_{m \in {\cal S}} \max_{n \in {\cal S}^c} \left \{  K_{ X_m | X_n } \right \} I_{|{\cal S}|} \\
& = \sum_{m \in {\cal S}} \left( 1 - \min_{n \neq m} \rho_{mn}^2 \right) I_{|{\cal S}|}  \overset{(b)}{\preceq} |{\cal S}| \left( 1 - \rho^2 \right) I_{|{\cal S}|}
\end{align*}
where $(a)$ follows from the fact that $K_{ X_m | {\cal S}^c } \leq   K_{ X_m | X_n }$ for some $n \in {\cal S}^c$; and $(b)$ is because we set $\rho = \min_{m,n} \rho_{mn}$.
\end{proof}

Similarly we define $\delta_{\sf sum} = (\ref{eq:outerbound12_Mtx}) - (\ref{eq:achievablerate12_Mtx})$. We then get
\begin{align*}
\delta_{\sf sum} & \leq \log \left ( 1 +  \left\{ \sum_{m=1}^{M} \sqrt{ {\sf SNR}_{mk} } \right \}^2 \right ) - \log \left (1 + \sum_{m=1}^{M} {\sf SNR}_{mk}   \right ) \\
&\overset{(a)}{\leq}\log \left ( 1 +   M  \sum_{m=1}^{M}  {\sf SNR}_{mk}   \right ) - \log \left (1 + \sum_{m=1}^{M} {\sf SNR}_{mk}   \right ) \\
& {\leq} \log M  \\
\end{align*}
where $(a)$ follows from the Cauchy-Schwarz inequality. Therefore, the gap per transmitter is upper-bounded by
\begin{align}
\frac{\delta_{\sf sum}}{M} \leq \frac{ \log M}{M} < \log \left\{ 2 (M-1) \right \}.
\end{align}
This completes the proof.

\section{Converse Proof of Theorem~\ref{theorem:ADT}}
\label{appendix:converseproof_ADT}

First, consider (\ref{eq:ADTbound2}). Starting with Fano's inequality, we get
\begin{align*}
\begin{split}
N& \left( \sum_{m=1}^{M} R_m - \epsilon_N \right) \leq I(W_1,\cdots, W_M;Y_k^{N}) \\
& \overset{(a)}{\leq} \sum  H(Y_{ki}) \overset{(b)}{\leq} N \max \{ n_{1k}, \cdots, n_{Mk} \}
\end{split}
\end{align*}
where $(a)$ follows from the fact that conditioning reduces entropy; and $H(Y_{ki})$ is maximized when $(X_{1i},\cdots,X_{Mi})$ are uniformly distributed and independent.

Next, consider (\ref{eq:ADTbound1}). Let $W_{\cal S} := \left\{ W_m: m \in {\cal S}\right \}$. Starting with Fano's inequality, we get
\begin{align*}
\begin{split}
N& \left( \sum_{m \in {\cal S}} R_m - \epsilon_N \right)  \leq I(W_{\cal S};Y^N, W_{{\cal S}^c}) \\
&\overset{(a)}= I(W_{\cal S};Y^N | W_{{\cal S}^c})  \overset{(b)}= \sum H(Y_{i}| W_{{\cal S}^c} ,Y^{i-1},X_{ {\cal S}^c }^{i}  ) \\
&\overset{(c)}{\leq} \sum H(Y_{i}| X_{ {\cal S}^c,i }  ) \overset{(d)}{\leq} N {\sf rank} (G_{\cal S} ) \\
\end{split}
\end{align*}
where ($a$) follows from the fact that $W_{\cal S}$ and $W_{{\cal S}^c}$ are independent;
$(b)$ follows from the fact that $X_{ {\cal S}^c }^{i}$ is a function of $(W_{{\cal S}^c} ,Y_1^{i-1},\cdots, Y_K^{i-1})$; ($c$) follows from the fact that conditioning reduce entropy; $(d)$ follows from the fact that $H(Y_{i}| X_{ {\cal S}^c,i }  )$ is maximized when $(X_{1i},\cdots,X_{Mi})$ are uniformly distributed and independent.

\bibliographystyle{ieeetr}
\bibliography{bib_multicastfeedback}

\begin{thebibliography}{10}

\bibitem{SK:it}
J.~P.~M. Schalkwijk and T.~Kailath, ``A coding scheme for additive noise
  channels with feedback - part {I}: No bandwith constraint,'' {\em IEEE
  Transactions on Information Theory}, vol.~12, pp.~172--182, Apr. 1966.

\bibitem{shannon:it}
C.~E. Shannon, ``The zero error capacity of a noisy channel,'' {\em IRE
  Transactions on Information Theory}, vol.~2, pp.~8--19, Sept. 1956.

\bibitem{Gaarder:it}
N.~T. Gaarder and J.~K. Wolf, ``The capacity region of a multiple-access
  discrete memoryless channel can increase with feedback,'' {\em IEEE
  Transactions on Information Theory}, Jan. 1975.

\bibitem{Ozarow:it}
L.~H. Ozarow, ``The capacity of the white {G}aussian multiple access channel
  with feedback,'' {\em IEEE Transactions on Information Theory}, vol.~30,
  pp.~623--629, July 1984.

\bibitem{Kramer:it02}
G.~Kramer, ``Feedback strategies for white {G}aussian interference networks,''
  {\em IEEE Transactions on Information Theory}, vol.~48, pp.~1423--1438, June
  2002.

\bibitem{SuhTse}
C.~Suh and D.~Tse, ``Feedback capacity of the {G}aussian interference channel
  to within 2 bits,'' {\em IEEE Transactions on Information Theory}, vol.~57,
  pp.~2667--2685, May 2011.

\bibitem{GastparAmosYossefWigger}
M.~Gastpar, A.~Lapidoth, Y.~Steinberg, and M.~Wigger, ``Feedback can double the
  prelog of some memoryless {G}aussian networks,'' {\em arXiv:1003.6082}, Jan.
  2012.

\bibitem{Maric:isit05}
I.~Maric, R.~D. Yates, and G.~Kramer, ``The discrete memoryless compound
  multiple access channels with conference encoders,'' {\em IEEE International
  Symposium on Information Theory}, Sept. 2005.

\bibitem{Simeone:it09}
O.~Simeone, D.~Gunduz, H.~V. Poor, A.~J. Goldsmith, and S.~Shamai, ``Compound
  multiple-access channels with partial cooperation,'' {\em IEEE Transaction on
  Information Theory}, vol.~55, pp.~2425--2441, June 2009.

\bibitem{Cover:it81}
T.~M. Cover and C.~S.~K. Leung, ``An achievable rate region for the
  multiple-access channel with feedback,'' {\em IEEE Transactions on
  Information Theory}, vol.~27, pp.~292--298, May 1981.

\bibitem{Willems:it82}
F.~M.~J. Willems, ``The feedback capacity region of a class of discrete
  memoryless multiple access channels,'' {\em IEEE Transactions on Information
  Theory}, vol.~28, pp.~93--95, Jan. 1982.

\bibitem{BrossLapidoth:it05}
S.~I. Bross and A.~Lapidoth, ``An improved achievable rate region for the
  discrete memoryless two-user multiple-access channel with noiseless
  feedback,'' {\em IEEE Transactions on Information Theory}, vol.~51,
  pp.~811--833, Mar. 2005.

\bibitem{Venkataramanan:it11}
R.~Venkataramanan and S.~S. Pradhan, ``A new achievable rate region for the
  multiple-access channel with noiseless feedback,'' {\em IEEE Transactions on
  Information Theory}, vol.~57, pp.~8038--8054, Dec. 2011.

\bibitem{Willems:it83}
F.~M.~J. Willems, ``The discrete memoryless multiple access channel with
  partially cooperating encoders,'' {\em IEEE Transactions on Information
  Theory}, vol.~29, pp.~441--445, May 1983.

\bibitem{LimKimElGamalChung:it11}
S.~H. Lim, Y.-H. Kim, A.~El-Gamal, and S.-Y. Chung, ``Noisy network coding,''
  {\em IEEE Transaction on Information Theory}, vol.~57, pp.~3132--3152, May
  2011.

\bibitem{Salman:IT11}
S.~Avestimehr, S.~Diggavi, and D.~Tse, ``Wireless network information flow: A
  deterministic approach,'' {\em IEEE Transactions on Information Theory},
  vol.~57, pp.~1872--1905, Apr. 2011.

\bibitem{Willems:it85}
F.~M.~J. Willems and E.~C. van~der Meulen, ``The discrete memoryless
  multiple-access channel with cribbing encoders,'' {\em IEEE Transactions on
  Information Theory}, vol.~31, pp.~313--327, May 1985.

\bibitem{Massimo:LQG}
E.~Ardestanizadeh, P.~Minero, and M.~Franceschetti, ``{LQG} control approach to
  {G}aussian broadcast channels with feedback,'' {\em submitted to the IEEE
  Transactions on Information Theory (arXiv:1102.3214)}, Feb. 2011.

\bibitem{AlirezaSuhAves}
A.~Vahid, C.~Suh, and A.~S. Avestimehr, ``Interference channels with
  rate-limited feedback,'' {\em IEEE Transactions on Information Theory},
  vol.~58, pp.~2788--2812, May 2012.

\bibitem{SuhTse:isit12}
C.~Suh, I.-H. Wang, and D.~Tse, ``Two-way interference channels,'' {\em
  Proceedings of the IEEE International Symposium on Information Theory, MIT,
  USA}, July 2012.

\bibitem{ElGamalKim:Book}
A.~E. Gamal and Y.-H. Kim, {\em Network Information Theory}.
\newblock New York: Cambridge University Press, 2011.

\bibitem{CoverThomas}
T.~M. Cover and J.~A. Thomas, {\em Elements of Information Theory}.
\newblock New York Wiley, 2th~ed., July 2006.

\end{thebibliography}

\end{document}